\documentclass[11pt,a4paper]{article}
\pdfoutput=1
\usepackage[utf8]{inputenc}

\usepackage{amsmath}
\usepackage{amsfonts}
\usepackage{dsfont}
\usepackage{mathtools}   
\usepackage{empheq}
\usepackage{ifthen}

\usepackage{nicefrac}
\usepackage{comment}
\usepackage{enumerate}

\usepackage{myprobdescr}

\usepackage[noend]{algpseudocode}
\usepackage{algorithm}
\usepackage{authblk}

\usepackage{scrpage2}

\usepackage{amsthm}
\newtheorem{theorem}{Theorem}{\bfseries}{\normalfont}
\newtheorem{lemma}{Lemma}{\bfseries}{\normalfont}
\newtheorem{definition}{Definition}{\bfseries}{\normalfont}
\newtheorem{proposition}{Proposition}{\bfseries}{\normalfont}
\newtheorem{rrule}{Rule}{\bfseries}{\normalfont} 
\newtheorem{corollary}{Corollary}{\bfseries}{\normalfont}
{\bfseries}{\normalfont}
\usepackage{tikz}
\usetikzlibrary{positioning,backgrounds,patterns,calc,snakes}

\usepackage{bbding}
\usepackage{placeins}

\newcommand{\probDefEnv}[1]{%
  \begin{center}%
    \begin{minipage}{0.9\linewidth}%
		#1
    \end{minipage}%
  \end{center}%
}

\usepackage{microtype,ellipsis}

\usepackage[pagebackref,pdfdisplaydoctitle,menucolor=orange!40!black,filecolor=magenta!40!black,urlcolor=blue!40!black,linkcolor=red!40!black,citecolor=green!40!black,colorlinks]{hyperref}

\hypersetup{pdftitle={Win-Win Kernelization}, pdfauthor={Vincent Froese and André Nichterlein and Rolf Niedermeier}}

\usepackage[sort&compress]{cleveref}
\crefname{rrule}{Reduction Rule}{Reduction Rules}
\crefname{chapter}{Chapter}{Chapters}
\crefname{section}{Section}{Sections}
\crefname{subsection}{Section}{Sections}
\crefname{theorem}{Theorem}{Theorems}
\crefname{obs}{Observation}{Observations}
\crefname{proposition}{Proposition}{Propositions}
\crefname{corollary}{Corollary}{Corollaries}
\crefname{lemma}{Lemma}{Lemmata}
\crefname{equation}{Equivalence}{Equivalences}
\crefname{figure}{Figure}{Figure}

\usepackage[numbers,sort]{natbib}

\makeatletter
\def\NAT@spacechar{~}
\makeatother

\DeclareRobustCommand{\NoKernelAssume}{$\text{NP}\subseteq \text{{coNP/poly}}$}

\newcommand{\N}{\mathds{N}}

\newcommand{\kDegAnon}{\textsc{Degree Anonymity}\xspace}

\newcommand{\DCE}[1][$S$]{\textsc{DCE}\textnormal{(#1)}\xspace}
\newcommand{\NCE}{\textsc{NCE}\xspace}
\newcommand{\clique}{\textsc{Clique}\xspace}
\newcommand{\IS}{\textsc{Independent Set}\xspace}
\newcommand{\VC}{\textsc{Vertex Cover}\xspace}
\newcommand{\PiEA}{$\Pi$\nobreakdash-\textsc{DSC}\xspace}
\newcommand{\PiNA}{$\Pi$\nobreakdash-\textsc{NSC}\xspace}
\newcommand{\PiDec}{$\Pi$\nobreakdash-\textsc{Decision}\xspace}
\newcommand{\degList}{\ensuremath{\tau}}
\newcommand{\numList}{\ensuremath{\phi}}

\newcommand{\fFactor}{\textsc{$f$\nobreakdash-Factor}\xspace}
\newcommand{\ffactor}[1][f]{\ensuremath{#1}\nobreakdash-factor\xspace}

\newcommand{\factorSizeBound}{\ensuremath{(r+1)^2}}
\newcommand{\realizationSizeBound}{\ensuremath{r(r+1)^2}}

\newcommand{\blockSet}[1][\alpha]{\ensuremath{#1}-block set\xspace}
\newcommand{\typeSet}[1][\alpha]{\ensuremath{#1}-type set\xspace}

\begin{document}


\title{Win-Win Kernelization for \\ Degree Sequence Completion Problems}
\author[]{Vincent Froese\footnote{Supported by Deutsche Forschungsgemeinschaft, project DAMM (NI 369/13).}}
\author[]{André Nichterlein}
\author[]{Rolf Niedermeier}

\affil[]{Institut f\"ur Softwaretechnik und Theoretische Informatik,
  TU Berlin, Germany, \texttt{\{vincent.froese, andre.nichterlein, rolf.niedermeier\}@tu-berlin.de}}
\date{}

\maketitle
\thispagestyle{scrheadings}
\cfoot{}
\ohead{}
\ifoot{To appear in \emph{Journal of Computer and System Sciences}, 2016. A~manuscript of this article has been published in \emph{Proceedings of the 14th Scandinavian Symposium and Workshops on Algorithm Theory (SWAT' 14), Volume 8503 of LNCS, pp. 194--205, 2014. \copyright~Springer.}}
\begin{abstract}
We study provably effective and
efficient data reduction for a class of NP-hard graph modification problems based on vertex degree 
properties. 
We show fixed-parameter tractability for NP-hard graph completion (that is, edge addition) cases while we show that there is no hope to achieve analogous results for the corresponding vertex or edge deletion versions. 
Our algorithms are based on transforming graph completion problems into efficiently solvable number problems and exploiting $f$-factor computations for translating the results back into the graph setting.
Our core observation is that we encounter a win-win situation: either the number of edge additions is small or the problem is polynomial-time solvable. 
This approach helps in answering an open question by Mathieson and Szeider~[\emph{JCSS}~2012] concerning the polynomial kernelizability of \textsc{Degree Constraint Edge Addition} and leads to a general method of approaching polynomial-time preprocessing for a wider class of degree sequence completion problems.
\end{abstract}

\section{Introduction}\label{sec:Intro}
We propose a general approach for achieving
polynomial-size problem kernels for a class of graph completion
problems where the goal graph has to fulfill certain degree
properties. Thus, we explore and enlarge results on provably effective
polynomial-time preprocessing for these NP-hard graph problems.
To a large extent, the initial motivation for our work comes from
studying the NP-hard 
graph modification problem \textsc{Degree Constraint
Editing($S$)} for non-empty subsets~$S \subseteq
\{\text{v}^-,\text{e}^+,\text{e}^-\}$ of editing operations
(v$^-$: ``vertex deletion'', e$^+$: ``edge addition'', e$^-$: ``edge
deletion'') as introduced by \citet{MS12}.\footnote{\citet{MS12}
originally introduced a weighted version of the problem,
where the vertices and edges can have positive integer weights
incurring a cost for each editing operation. Here, we focus on
the unweighted version.}
The definition reads as follows.

\probDefEnv{\defDecprob{Degree Constraint Editing($S$) (\DCE)}
{An undirected graph~$G=(V,E)$, two integers~$k, r >0$, and a ``degree list function''~$\degList\colon V \rightarrow 2^{\{0,\dotsc,r\}}$.}
{Is it possible to obtain a graph $G'=(V',E')$ from $G$ using at most~$k$ editing operations of type(s) as specified by~$S$ such that~$\deg_{G'}(v) \in \degList(v)$ for all~$v \in V'$?}}

\noindent In our work, the set~$S$ always consists of a single editing
operation.
Our studies focus on the two most obvious parameters: the
number~$k$ of editing operations and the maximum allowed degree~$r$.
We will show that, although all three variants are
NP-hard, \DCE[e$^+$] is amenable to a generic kernelization method we propose.
This method is based on dynamic programming solving a corresponding
number problem and $f$-factor computations. For \DCE[e$^-$] and
\DCE[v$^-$], however, we show that there is little hope to achieve
analogous results.

\paragraph{Previous Work}
There are basically two fundamental starting points for our work.
First, there is our previous theoretical work on degree
anonymization\footnote{For a given integer~$k$, the task here is to add as few edges as possible to a graph such that the resulting graph is \emph{$k$-anonymous}, that is, there is no vertex degree in the graph which occurs at least one but less than $k$~times.} 
in social networks~\cite{HNNS15} motivated and strongly inspired by a preceding heuristic approach due to~\citet{LT08} (also see \citet{CLT10} for an extended version). 
Indeed, our previous work for degree anonymization 
inspired empirical work with encouraging experimental results~\cite{HHN14}.
A fundamental contribution of this work now is to systematically reveal what the problem-specific parts (tailored towards degree anonymization) and what the ``more general'' parts of that approach are. 
In this way, we develop this approach into a general method of wider applicability for a number of graph completion problems based on degree properties. 
The second fundamental starting point is Mathieson and Szeider's work~\cite{MS12} on (weighted) \DCE[$S$]. 
They showed several exponential-size problem kernels for the operations vertex deletion and edge deletion.
For the case that the degree list of each vertex contains only one number, they even obtain
a polynomial-size problem kernel.
They left open, however, whether it is possible to reduce
\DCE[e$^+$] in polynomial time to a problem kernel of size polynomial in~$r$---we will affirmatively answer this question.
Indeed, while Mathieson and Szeider provided several fixed-parameter
tractability results with respect to the combined parameter $k$
and~$r$, they partially left open whether similar results can be
achieved using the stronger parameterization\footnote{Fixed-parameter
  tractability with respect to the parameter~$r$ (trivially) implies
  fixed-parameter tractability with respect to the combined
  parameter~$(k,r)$, but the reverse clearly does not hold in general;
  see \citet{KN12} for a broader discussion in this direction.} with the single parameter~$r$. 
Recently, \citet{Gol14} described kernelization and fixed-parameter results for closely related graph editing problems where vertex and edge deletions and edge insertions are allowed, the degree list of each vertex contains exactly one number, and the resulting graph has to be connected. 

From a more general perspective, all these considerations fall into the category of ``graph editing to fulfill degree constraints'', which recently received significant interest in terms of parameterized complexity analysis~\cite{BBHNW16,FGMN11,Gol14,MT09}. 

\medskip
\noindent\textit{Our Contributions.}
Answering an open question of Mathieson and Szeider~\cite{MS12}, we present
an~$O(kr^2)$-vertex kernel for \DCE[e$^+$] which we then
transfer into an~$O(r^5)$-vertex kernel using a strategy
rooted in previous work~\cite{LT08,HNNS15}. A further main contribution 
of our work in the spirit of meta kernelization \cite{BFLPST09} is to clearly separate problem-specific from problem-independent 
aspects of this strategy, thus making it accessible to a wider 
class of degree sequence completion problems. 
We observe that if the goal graph shall have ``small'' maximum degree~$r$,
then the actual graph structure is in a sense negligible and  thus 
allows for a lot of freedom
that can be algorithmically exploited. 
This paves the way to a \emph{win-win situation} of either having guaranteed a small number of edge additions or the overall problem being solvable in polynomial-time anyway---another example in the list of win-win situations exploited in parameterized algorithmics~\cite{Fel03WG}.

Besides our positive kernelization results, we exclude polynomial-size
problem kernels for \DCE[e$^-$] and \DCE[v$^-$] subject to the
assumption that~$\text{NP}\not\subseteq \text{{coNP/poly}}$, thereby showing that
the exponential-size kernel results by \citet{MS12} are essentially tight.
In other words, this demonstrates that in our context edge completion is much more 
amenable to kernelization than edge deletion or vertex deletion are.
We also prove NP-hardness of \DCE[v$^-$] and \DCE[e$^+$] for graphs of
maximum degree three, implying that the maximum degree is not a useful
parameter for parameterized complexity or kernelization purposes.
Last but not least, we develop a general preprocessing approach for
\textsc{Degree Sequence Completion} problems which yields a search
space size that is polynomially bounded in the parameter.
While this per se does not give polynomial kernels, we derive fixed-parameter tractability with respect to the combined parameter maximum degree and solution size. 
The usefulness of our method is illustrated by further example degree sequence completion problems.

\medskip
\noindent\textit{Notation.}
All graphs in this paper are undirected, loopless, and simple (that is, without multiple edges).
For a graph~$G = (V,E)$, we set~$n := |V|$ and~$m := |E|$.
The degree of a vertex~$v \in V$ is denoted by~$\deg_G(v)$, the \emph{maximum vertex degree} by $\Delta_G$, and the \emph{minimum vertex degree} by~$\delta_G$.
For a finite set~$U$, we denote by~$\binom{U}{2}$ the set of all size-two subsets of~$U$.
We denote by~$\overline{G} := (V,{V \choose 2} \setminus E)$ the \emph{complement graph} of~$G$.
For a vertex subset~$V' \subseteq V$, the subgraph induced by~$V'$ is denoted by~$G[V']$.
For an edge subset~$E'\subseteq {V \choose 2}$, $V(E')$ denotes the set of all endpoints of edges in~$E'$ and $G[E']:=(V(E'),E')$.
For a set~$E'$ of edges with endpoints in a graph~$G$, we denote by~$G+E' := (V, E \cup E')$ the graph that results from inserting all edges in~$E'$ into~$G$.
Similarly, we define for a vertex set~$V' \subseteq V$, the graph~$G - V' := G[V \setminus V']$.
For each vertex $v\in V$, we denote by $N_G(v):=\{u\in V\mid \{u,v\}\in E\}$ the \emph{open neighborhood} of~$v$ in~$G$ and by $N_G[v]:=N_G(v)\cup \{v\}$ the \emph{closed neighborhood}. 
We omit subscripts if the corresponding graph is clear from the context.
A vertex~$v \in V$ with $\deg(v) \in \degList(v)$ is called \emph{satisfied} (otherwise \emph{unsatisfied}).
We denote by~$\mathcal{U} \subseteq V$ the set of all unsatisfied vertices, formally $\mathcal{U} := \{v \in V \mid \deg_G(v) \notin \degList(v)\}$.

\medskip
\noindent\textit{Parameterized Complexity.}
This is a two-di\-men\-sional framework for studying computational complexity~\citep{DF13,FG06,Nie06}.
One di\-men\-sion of a parameterized problem is the input size~$s$, and the other one is the \emph{parameter} (usually a positive integer).
A parameterized problem is called \emph{fixed-parameter tractable} (fpt) with respect to a parameter~$\ell$ if it can be solved in~$f(\ell)\cdot s^{O(1)}$ time, where~$f$ is a computable function only depending on~$\ell$.
This definition also extends to \emph{combined parameters}. 
Here, the parameter usually consists of a tuple of positive integers~$(\ell_1, \ell_2, \ldots)$ and a parameterized problem is called fpt with respect to~$(\ell_1, \ell_2, \ldots)$ if it can be solved in~$f(\ell_1, \ell_2, \ldots)\cdot s^{O(1)}$ time.

A core tool in the development of fixed-parameter algorithms is
polynomial-time preprocessing by \emph{data reduction}~\citep{Kra14,GN07a}. 
Here, the goal is to transform a given problem instance~$I$ with parameter~$\ell$ in polynomial time into an equivalent instance~$I'$ with parameter~$\ell'\leq g(\ell)$ for some function~$g$ such that the size of~$I'$ is upper-bounded by some function~$h$ depending only on~$\ell$.
If this is the case, we call~$I'$ a (problem) \emph{kernel} of size~$h(\ell)$.
If~$h$ is a polynomial, then we speak of a \emph{polynomial kernel}.
Usually, this is achieved by applying polynomial-time computable data reduction rules.  
We call a data reduction rule~$\mathcal{R}$ \emph{correct} if the new instance~$I'$ that results from applying~$\mathcal{R}$ to~$I$ is a yes-instance if and only if~$I$ is a yes-instance.  
The whole process is called \emph{kernelization}. 
It is well known that a parameterized problem is fixed-parameter tractable if and only if it has a problem kernel \cite{CCDF97}.

\section{Degree Constraint Editing}
\citet{MS12} showed fixed-parameter tractability for \DCE for all
non-empty subsets~$S \subseteq \{v^-, e^-, e^+\}$ with respect to the
combined parameter~$(k,r)$ and W[1]-hardness with respect to the
single parameter~$k$.
The fixed-parameter tractability is in a sense tight as \citet{Cor88}
proved that \DCE[e$^-$] is NP-hard on planar graphs with maximum
degree three and with maximum allowed degree~$r=3$ (that is, it is
paraNP-hard with respect to~$r$, meaning NP-hard for a constant
parameter value~$r$) and thus presumably not fixed-parameter
tractable with respect to~$r$ (unless P=NP).
We complement his result by showing that \DCE[v$^-$] is NP-hard on cubic (that is three-regular) planar graphs, even if~$r=0$, and that \DCE[e$^+$] is NP-hard on planar graphs with maximum degree three.
(Note that in this case, even though the input graph is planar,
the graph obtained from edge additions does not necessarily have to be planar.)

\begin{proposition}\label{thm:DCEv-NPhard}
	\DCE[v$^-$] is NP-hard on cubic planar graphs, even if~$r=0$.
\end{proposition}
\begin{proof}
	We provide a polynomial-time many-one reduction from the NP-hard \VC on cubic planar graphs~\cite{GJ79}.
	Given a cubic graph~$G = (V,E)$ and a positive integer~$h$, the \VC problem asks for a subset~$V'$ of at most~$h$ vertices such that each edge in~$E$ has at least one endpoint in~$V'$.
	Let~$I = (G = (V,E),h)$ be a \VC instance with the cubic planar graph~$G$.
	Now consider the \DCE[v$^-$] instance $I' = (G,h,0,\degList)$ with~$\degList(v) = \{0\}$ for all~$v\in V$.
	Observe that~$I$ and~$I'$ are yes-instances if and only if at most~$h$ vertices can be removed from~$G$ such that the resulting graph is edgeless.
	Hence, $I$ is a yes-instance if and only if~$I'$ is a yes-instance.
\end{proof}

\begin{proposition}\label{thm:DCEe+NPhard}
	\DCE[e$^+$] is NP-hard on planar graphs with maximum degree three.
\end{proposition}
\begin{proof}
	We provide a polynomial-time many-one reduction from the NP-hard \IS problem on cubic planar graphs~\cite{GJ79}.
	Given a cubic graph~$G = (V,E)$ and a positive integer~$h$, the \IS problem asks for a subset~$V'$ of at least~$h$ pairwise non-adjacent vertices.
	Given an \IS instance~$(G = (V,E), h)$, we construct a \DCE[e$^+$] instance~$(G',k,h,\degList)$ as follows:
	Start with~$G'$ as a copy of~$G$, add a new vertex~$v$ to~$G'$ and set~$\degList(v) := \{h\}$. 
	Furthermore, for all other vertices~$u \in V$ set~$\degList(u) = \{3, 3+h\}$.
	Finally, set~$k := \binom{h}{2}+h$.
	It is straightforward to argue that the only way of satisfying~$v$ within the given budget is to connect it to~$h$ vertices forming an independent set.
\end{proof}

We remark that in \cref{thm:DCEe+NPhard} we do \emph{not} require that the output graph is planar.
In contrast to \DCE[e$^-$] and \DCE[v$^-$], unless P${}={}$NP, \DCE[e$^+$] cannot be NP-hard for constant values of~$r$ since we will later show fixed-parameter tractability for \DCE[e$^+$] with respect to the parameter~$r$.

\subsection{Excluding Polynomial-Size Problem Kernels}
\citet{MS12} gave exponential-size problem kernels for \DCE[v$^-$] and \DCE[\{v$^-$, e$^-$\}] with respect to the combined parameter~$(k,r)$.
We prove that these results are tight in the sense that, under standard complexity-theoretic assumptions, neither \DCE[e$^-$] nor \DCE[v$^-$] admits a polynomial-size problem kernel when
parameterized by~$(k,r)$.
Note that \citet{Gol15} showed that the problem variant \DCE[v$^-$, e$^+$] (vertex deletions and edge insertions are allowed) where the degree list of each vertex contains exactly one number does not admit a polynomial-size kernel unless $\text{NP}\not\subseteq \text{{coNP/poly}}$.


\begin{theorem} \label{thm:no-poly-kernel-DCEe}
	\DCE[e$^-$] does not admit a polynomial-size problem kernel with respect to~$(k,r)$ unless \NoKernelAssume.
\end{theorem}
\begin{proof}
	We provide a polynomial time and parameter transformation from the \clique problem parameterized by the ``vertex cover number''.
	Given a graph~$G = (V,E)$ and a positive integer~$h$, the \clique problem asks for a subset of at least~$h$ vertices that are pairwise adjacent. 
	A vertex cover for a graph~$G$ is a vertex subset~$V'$ such that each edge in~$E$ has at least one endpoint in~$V'$.
	The \emph{vertex cover number} is the size of a minimum vertex cover.
	Given a parameterized problem instance $(I,k)$, a \emph{polynomial time and parameter transformation} yields an equivalent instance~$(I',p(k))$ in time~$q(|I|)$ for polynomials~$p$ and~$q$~\cite{BTY11}.
	Since \clique parameterized by the vertex cover number does not admit a polynomial-size problem kernel~\cite{BJK14} and \clique as well as \DCE[e$^-$] are NP-complete, it then follows that also \DCE[e$^-$] does not admit a polynomial-size problem kernel with respect to~$(k,r)$~\cite[Theorem 2.15]{BJK14}.

	The details of the transformation are as follows.
	Let~$(G=(V,E),h)$ be the \clique instance and let~$X \subseteq V$ be a factor-2 approximation of a minimum vertex cover of~$G$ 
	(efficiently computable by finding a maximal matching).
	We assume without loss of generality that each vertex in~$G$ has degree at least~$h$.
	Note that any clique in~$G$ has size at most $|X|+1$ since $V\setminus X$ forms an independent set of which at most one vertex can be in a clique. 
	We can thus assume that~$h \le |X| + 1$.
	Two vertices~$v_1,v_2 \in V \setminus X$ are called \emph{twins} (with respect to~$X$) if they have the same neighbors in~$X$.
	A \emph{twin class} with respect to~$X$ is a maximal subset~$C \subseteq V \setminus X$ such that all pairs of vertices in~$C$ are twins. 
	Obviously, the twin classes provide a partition of~$V \setminus X$.
	Let~$C_1, \ldots, C_\ell$ be the twin classes with respect to~$X$. 
	It holds that $\ell \le \min\{2^{|X|},n\}$.
	The idea of the construction is to have $\ell$~disjoint copies of~$G[X]$, one together with a representative vertex of each twin class. 
	The degree lists and the budget of edge deletions are chosen such that if we delete any edge from one of the copied subgraphs, then it is only possible to delete all edges of a clique of size~$h$ in that subgraph.
	Moreover, the subgraphs are connected to a binary tree which serves as a ``selector'' forcing any solution to delete edges in exactly one of the subgraphs.

	The precise construction of the \DCE[e$^-$] instance~$(G' = (V', E'),k,r,\degList)$ works as follows:
	For~$i \in \{1, \ldots, \ell\}$, let~$v_i$ be an arbitrary vertex of the twin class~$C_i$.
	Furthermore, let~$G_i := G[X \cup \{v_i\}]$.
	Initialize~$G'$ as the disjoint union of the graphs~$G_i$, $1\le i \le \ell$.
	Observe that~$G'$ contains~$\ell$ copies of the vertex cover~$X$.
	Next, add a binary tree of height~$t:=\lceil\log \ell \rceil\le |X|$ with leaves~$u_1, \ldots, u_{2^t}$ to~$G'$.
	For each~$i \in \{1,\ldots,\ell\}$, make~$u_i$ adjacent to all vertices in~$G_i$.
        If $2^t >\ell$, then, for each $i\in\{\ell+1,\ldots,2^t\}$, add another copy of $G_1$ and make~$u_i$ adjacent
        to all vertices in this new copy.
	This completes the construction of~$G'$; see \cref{fig:no-poly-construction} for an example.
	\begin{figure}[t]
		\centering
		\def\layersep{1}
		\begin{tikzpicture}[draw=black!30, node distance=\layersep]
				\tikzstyle{knoten}=[circle,draw,fill=white,minimum size=15pt,inner sep=0pt]
				\tikzstyle{dummy}=[circle,draw,fill=white,minimum size=5pt,inner sep=0pt]
				
				\foreach \i in {1,2,3} {
					\node[knoten] (VC-\i) at (0,\i) {$v_\i$};
				}
				\node (VC) at (0,5) {VC};
				\node (IS) at (2,5) {IS};

				\foreach \i in {1,...,5} {
					\node[knoten] (IS-\i) at (2,\i-1) {$w_\i$};
				}

				\path (VC-2) edge[-] (VC-3);
				\foreach \i / \j in {1/1, 1/2, 1/3, 2/2, 2/3, 2/4, 2/5, 3/3, 3/4, 3/5} {
					\path (VC-\i) edge[-] (IS-\j);
				}

				\begin{pgfonlayer}{background}
					\path (IS.north -| VC-1.west)+(-0.25,0.25) node (a) {};
					\path (IS-1.south -| VC-1.east)+(+0.25,-0.25) node (b) {};
					\path[fill=gray!10,rounded corners, draw=black!50, dashed]
							(a) rectangle (b);

					\path (IS.north -| IS-1.west)+(-0.25,0.25) node (a) {};
					\path (IS-1.south -| IS-1.east)+(+0.25,-0.25) node (b) {};
					\path[fill=gray!10,rounded corners, draw=black!50, dashed]
							(a) rectangle (b);

					\foreach \i / \j in {1/1, 2/2, 3/3, 4/5} {
						\path (IS-\j.north -| IS-1.west)+(-0.1,0.1) node (a) {};
						\path (IS-\i.south -| IS-1.east)+(+0.1,-0.1) node (b) {};
						\path[fill=gray!4,rounded corners, draw=black!50, dashed]
								(a) rectangle (b);
					}
				\end{pgfonlayer}

				\begin{scope}[xshift=7cm,yshift=4cm]
					\tikzstyle{neuron}=[draw=black, fill=white, circle, inner sep=2pt]
					\tikzstyle{thickedge}=[draw=black, line width=1pt]

					\node () at (0,1) {constructed graph};
					\node[label=above:{$v$}, neuron] (v) at (0,0) {};
					\node[neuron] (v1) at (-2,-\layersep) {} edge (v);
					\node[neuron,label=above:{$u_1$}] (u1) at (-3,-2*\layersep) {} edge (v1);
					\node[neuron,label=above:{$u_2$}] (u2) at (-1,-2*\layersep) {} edge (v1);

					\node[neuron] (v2) at (2,-\layersep) {} edge (v);
					\node[neuron,label=above:{$u_3$}] (u3) at (1,-2*\layersep) {} edge (v2);
					\node[neuron,label=above:{$u_4$}] (u4) at (3,-2*\layersep) {} edge (v2);

					\foreach \i in {1,...,4} {
						\foreach \j in {1,...,3} {
							\node[neuron] (G\i\j) at (2*\i + 0.5*\j - 6,-2.9*\layersep) {};
							\draw (G\i\j) -- (u\i);
						}
						\node (G\i) at (2*\i - 5,-3.5*\layersep) {$G_\i$};
						\draw (G\i2) -- (G\i3);
					}
					\node[neuron] (C1) at (-3.5,-3.6*\layersep) {} edge (u1) edge (G11);
					\node[neuron] (C2) at (-1.4,-3.6*\layersep) {} edge (u2) edge (G21) edge (G22);
					\node[neuron] (C3) at (0.6,-3.6*\layersep) {} edge (u3) edge (G31) edge (G32) edge (G33);
					\node[neuron] (C4) at (3.4,-3.6*\layersep) {} edge (u4) edge (G42) edge (G43);
					
					\node[below=of G1,yshift=0.8cm] () {$C_1 = \{w_1\}$};
					\node[below=of G2,yshift=0.8cm] () {$C_2 = \{w_2\}$};
					\node[below=of G3,yshift=0.8cm] () {$C_3 = \{w_3\}$};
					\node[below=of G4,xshift=0.2cm,yshift=0.8cm] () {$C_4 = \{w_4,w_5\}$};

					\begin{pgfonlayer}{background}
						\foreach \i in {1,...,4} {
							\path (G\i1.north -| G\i1.west)+(-0.1,0.1) node (a) {};
							\path (C1.south -| G\i3.east)+(+0.1,-0.1) node (b) {};
							\path[fill=gray!10,rounded corners, draw=black!50, dashed]
									(a) rectangle (b);
						}
					\end{pgfonlayer}
				\end{scope}
		\end{tikzpicture}
		\caption{An example of the construction. The given graph is displayed at the left side with highlighted vertex cover (VC), independent set (IS), and twin classes in the independent set. The constructed graph is depicted on the right side. The three vertices in the upper half of each subgraph~$G_i$, $1 \le i \le 4$, correspond to the vertex cover~$v_1, v_2, v_3$. The fourth vertex in the lower half of each subgraph corresponds to a vertex from the twin class~$C_i$.}
		\label{fig:no-poly-construction}
	\end{figure}
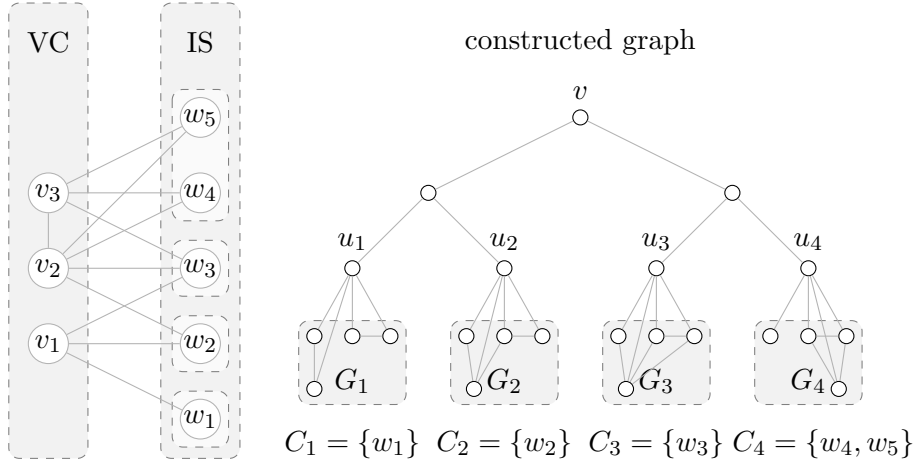
	Set~$k := \binom{h}{2} + h + t$.
	For each~$i \in \{1, \ldots, 2^t\}$ and each vertex~$v \in G_i$, set~$\degList(v) := \{\deg_{G'}(v), \deg_{G'}(v) - h\}$.
	For each leaf~$u_i$, set~$\degList(u_i) := \{\deg_{G'}(u_i), \deg_{G'}(u_i)-h-1\}$.
	For each inner vertex~$w$ in the binary tree, set~$\degList(w) := \{3, 1\}$.
	Finally, for the root~$v$ of the binary tree, set~$\degList(v) := \{1\}$.
	Observe that~$r := \max_{u \in V'}\max\degList(u) \le |X| + 2$ and~$k \in O(|X|^2)$. 
	Moreover, the above construction can be done in polynomial time.

	It remains to show the correctness of our construction, that is, $I := (G,h)$ is a yes-instance of \clique if and only if~$I' := (G',k,r,\degList)$ is a yes-instance of \DCE[e$^-$].
        
	``$\Rightarrow:$'' 
	Let~$(G,h)$ be a yes-instance, that is, $G$ contains a clique~$C$ of size~$h$. 
	If~$C$ contains a vertex from~$V \setminus X$, then let~$C_i$ denote its twin class (recall that there is at most one such vertex in~$C$). 
	Otherwise, set $i:=1$. 
	Now, let~$E'' \subseteq E'$ contain all edges between copies of the clique vertices in~$G_i$. 
	Moreover, let~$E''$ contain all edges between~$u_i$ and the clique vertices in~$G_i$ and all edges along the path from~$u_i$ to the root~$v$. 
	The overall number of edges in~$E''$ is at most~$k$ and it can easily be verified that, for each $v \in V'$, it holds~$\deg_{G'-E''}(v) \in \degList(v)$. 
	Thus,~$(G',k,r,\degList)$ is a yes-instance.

	``$\Leftarrow:$'' 
	Let~$(G',k,r,\degList)$ be a yes-instance. 
	Note that since~$\deg_{G'}(v) = 2$ and~$\degList(v) = \{1\}$, one of the two edges incident to~$v$ has to be deleted. 
	Moreover, for all inner nodes~$w$ of the binary tree, we have~$\degList(w)=\{3,1\}$; this ensures that any solution deletes either zero or two edges incident to~$w$. 
	Consequently, every solution deletes at least all edges on one particular path from the root~$v$ to some leaf~$u_i$. 
	This requires~$t$ edge deletions. 
	Now consider the leaves of the binary tree. 
	Their degree constraints are chosen in such a way that any solution either deletes no edges or exactly~$h+1$ edges incident to a leaf vertex. 
	Thus, for the leaf~$u_i$ with one removed incident edge, it holds that~$h$ further edges from~$u_i$ to~$G_i$ are deleted in a solution. 
	Finally, after this minimum number of $t+h$ edge deletions, we are left with a budget of~$\binom{h}{2}$ edge deletions in order to decrease the degrees of all~$h$ affected vertices in~$G_i$ by~$h-1$. 
	This is possible with exactly $\binom{h}{2}$ edge deletions if and only if they form a clique in~$G_i$, which in turn corresponds to a clique in~$G$.
\end{proof}
\begin{theorem}\label{thm:no-poly-kernel-DCEv}
	\DCE[v$^-$] does not admit a polynomial-size problem kernel with respect to~$(k,r)$ unless \NoKernelAssume.
\end{theorem}
\begin{proof}
	We adjust our construction from \cref{thm:no-poly-kernel-DCEe} as follows:
	In the binary tree connecting all subgraphs~$G_i$, make for each inner vertex~$w$ its two children adjacent.
	Furthermore, change the degree lists of all inner vertices from~$\{1,3\}$ to~$\{2,4\}$.
	The idea is that if a parent vertex is deleted, then one of its two children also has to be deleted in order to satisfy the remaining child vertex.
	To ensure that the root~$v$ with~$\deg(v)=2$ is deleted, set~$\degList(v) := \{3\}$.
	In this way, the selection of the subgraph~$G_i$ via the binary tree works as in the reduction for \DCE[e$^-$].
	As edges cannot be deleted any more, we also have to adjust our construction at the subgraphs~$G_i$.
	For each leaf~$u_i$ of the binary tree, remove the edges between~$u_i$ and the vertices in~$G_i$ and add a new vertex~$u_i'$ that is adjacent to~$u_i$ and all vertices in~$G_i$.
	Furthermore, add for each~$i \in \{1, \ldots, \ell\}$ a clique~$C_i$ with~$|X|^2$ vertices and make~$u_i'$ adjacent to all vertices in~$C_i$.
	For each vertex~$w \in C_i$, set~$\degList(w) := \{\deg(w)\}$.
	Furthermore, set~$\degList(u_i) := \{1,3\}$ and~$\degList(u_i') := \{\deg(u_i'), |X|^2 + h\}$.
	For each vertex~$w \in V(G_i)$, set~$\degList(w) := \{\deg(w), h\}$.
	Finally, set~$k := \lceil \log \ell \rceil + |X|+1 - h$.
	Observe that~$k \in O(|X|)$ and~$r \in O(|X|^2)$.
	This construction requires polynomial time.

	It remains to show the correctness of our construction, that is, $I := (G,h)$ is a yes-instance of \clique if and only if~$I' := (G',k,r,\degList)$ is a yes-instance of~\DCE[v$^-$].
	
	``$\Rightarrow:$'' Let~$C \subset V$ be a clique of size~$h$ in~$G$. 
	As~$C$ can contain at most one vertex from~$V \setminus X$, there is a subgraph~$G_i$ in~$G'$ such that the vertices corresponding to~$C$ are also contained in~$G_i$. 
	Now, remove all other vertices in~$G_i$ and all vertices on the shortest path from~$u_i$ to the root~$v$ of the binary tree.
	Overall, we removed at most~$|X|+1 - h + \lceil \log \ell \rceil = k$ vertices. 
	Furthermore, observe that in the remaining graph all vertices are satisfied, implying that~$I'$ is a yes-instance.
	
	``$\Leftarrow:$'' Assume that~$C \subseteq V'$ is a solution for $I'$, that is, each vertex in~$G'-C$ is satisfied and~$|C|\le k$.
	First, observe that the root~$v$ of the binary tree is contained in~$C$.
	We now show that we can assume that exactly one of the two children~$v_1, v_2$ of~$v$ is also contained in~$C$:
	Suppose that neither~$v_1$ nor~$v_2$ is contained in~$C$. 
	Hence, at least one child-vertex~$v_1'$ of~$v_1$ and at least one child-vertex~$v_2'$ of~$v_2$ are contained in~$C$ since otherwise~$v_1$ or~$v_2$ would not be satisfied in~$G-C$.
	Denote with~$v_1''$ the second child-vertex of~$v_1$.
	We create a solution~$C'$ for~$I'$ such that~$|C| \ge |C'|$ by setting~$C' := (C \cup \{v_1\}) \setminus \{v_2'\}$ and removing from~$C$ all vertices in the subtrees with root~$v_2$ or with root~$v_1''$, that is, all vertices that are in the same connected component with~$v_2$ or~$v_1''$ in~$G' - \{v,v_1,v_1'\}$.
	As every vertex except~$v$ is satisfied in~$G'$, $C$ is a solution for~$I'$, and~$v_2$ and~$v_1''$ are satisfied in~$G'-C'$, it follows that~$C'$ is also a solution for~$I'$.
	By iteratively applying this procedure to all inner vertices of the binary tree, we can assume that in this binary tree exactly the shortest path from~$v$ to one leaf, say~$u_i$, is contained in the solution~$C$.
	Since~$|C_i| > k$, $\degList(w) = \{\deg(w)\}$ for all~$w \in C_i$, and~$u_i'$ is adjacent to all vertices in~$C_i$, it follows that~$u_i' \notin C$.
	As~$u_i \in C$, this implies that all but~$h$ vertices in~$G_i$ are contained in~$C$.
	Since for each~$w$ of these~$h$ remaining vertices it holds $\degList(w) = \{\deg(w), h\}$, it follows that they form a clique of order~$h$.
	Thus, $I$ is a yes-instance.
\end{proof}
%
Having established computational lower bounds, we next show that in contrast to \DCE[e$^-$] and \DCE[v$^-$], \DCE[e$^+$] admits a polynomial kernel with respect to~$(k,r)$.

\subsection{A Polynomial Kernel for \DCE[e$^+$] with Respect to $(k,r)$}
\label{sec:polyKernel-rk}
In order to describe the kernelization for \DCE[e$^+$], we need some further notation: 
For~$i \in \{0,\dotsc,r\}$, a vertex~$v \in V$ is of \emph{type}~$i$ if and only if~$\deg(v) + i \in \degList(v)$, that is, $v$ can be satisfied by adding~$i$ edges to it.
The set of all vertices of type~$i$ is denoted by~$T_i$.
Observe that a vertex can be of multiple types, implying that for~$i\neq j$ the vertex sets~$T_i$ and~$T_j$ are not necessarily disjoint. 
Furthermore, note that the type-0 vertices are exactly the satisfied ones.
We remark that there are instances for \DCE[e$^+$] where we might have to add edges between two satisfied vertices (though this may seem counter-intuitive):
Consider, for example, a three-vertex graph without any edges, the degree list function values are~$\{2\}, \{0,2\}, \{0,2\}$, and~$k=3$. 
The two vertices with degree list~$\{0,2\}$ are satisfied. 
However, the only solution for this instance is to add \emph{all} edges.

Now, we describe our kernelization algorithm:
The basic strategy is to keep the unsatisfied vertices~$\mathcal{U}$ and ``enough'' arbitrary vertices of each type (from the satisfied vertices) and delete all other vertices. 
The idea behind the correctness is that the vertices in a solution are somehow ``interchangeable''.
If an unsatisfied vertex needs an edge to a satisfied vertex of type~$i$, then it is not important which satisfied type-$i$ vertex is used. 
We only have to take care not to ``reuse'' the satisfied vertices to
avoid the creation of multiple edges.

Next, we specify what we mean by ``enough'' vertices: 
The ``magic number'' is $\alpha := k(\Delta_G + 2)$.
This leads to the definition of \emph{\typeSet{s}}: 

\begin{definition}\label{def:alpha-type-set}
	An \typeSet~$C\subseteq V$ is a vertex subset containing all unsatisfied vertices~$\mathcal{U}$ and~$\min\{\alpha,|T_i \setminus \mathcal{U}|\}$ type-$i$ vertices from~$T_i \setminus \mathcal{U}$ for each~$i \in \{1,\ldots,r\}$.
\end{definition}

We will soon show that for any fixed \typeSet~$C$, deleting all
vertices in~$V\setminus C$ results in an equivalent instance.
However, deleting a vertex changes the degrees of its neighbors.
Thus, we also have to adjust their degree lists.
Formally, for a vertex subset~$V' \subseteq V$, we define~$\degList_{V'}\colon (V \setminus V') \to 2^{\{0,\dotsc,r\}}$, where for each~$u \in V \setminus V'$, we set
$$\degList_{V'}(u) := \{ d \in \N  \mid d + |N_G(u) \cap V'| \in \degList(u)\}.$$
Then, \emph{safely removing} a vertex set~$V' \subseteq V$ from the instance~$(G,k,r,\degList)$ means to replace the instance with~$(G - V',k,r,\degList_{V'})$, see \cref{fig:SafelyRemoveExample} for an example.
\begin{figure}[t]
	\begin{tikzpicture}[draw=black!80]
			\tikzstyle{knoten}=[circle,draw,minimum size=14pt,inner sep=2pt]

			\node[knoten,label=below:{$\{1,2\}$}] (V1) at (0,0) {$u$};
			\node[knoten,label=below:{$\{3\}$}] (V2) at (2,0) {$v$};
			\node[knoten,label=above:{$\{3\}$}] (V3) at (2,1) {$w$};
			\node[knoten,label=below:{$\{2\}$}] (V4) at (4,0) {$x$};

			\path (V1) edge[-] (V2);
			\path (V3) edge[-] (V4);
			\path (V2) edge[-] (V4);
			\path (V2) edge[-] (V3);

			\draw [-to,thick,snake=snake,segment amplitude=.4mm,segment length=2mm,line after snake=1mm] (5 , 0.5) -- (7 , 0.5)
				node [above=1mm,midway,text width=3cm,text centered] {safely remove~$\{x\}$};

			\node[knoten,label=below:{$\{1,2\}$}] (W1) at (8,0) {$u$};
			\node[knoten,label=below:{$\{2\}$}] (W2) at (10,0) {$v$};
			\node[knoten,label=above:{$\{2\}$}] (W3) at (10,1) {$w$};

			\path (W1) edge[-] (W2);
			\path (W2) edge[-] (W3);
	\end{tikzpicture}
	\caption{An example for safely removing a vertex from a graph. The sets next to the vertices denote the degree lists defined by the degree list function~$\degList$. Observe that in both graphs the vertex~$u$ is of type zero and of type one, the vertex~$v$ is of type zero, and the vertex~$w$ is of type one.}
	\label{fig:SafelyRemoveExample}
\end{figure}
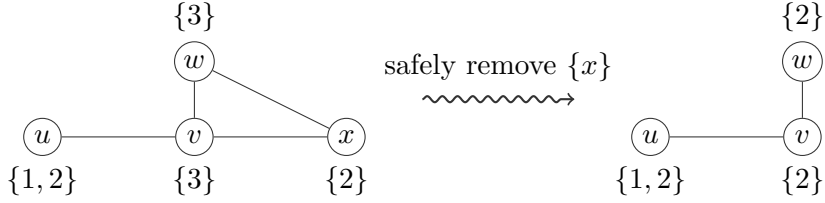
With these definitions we can provide our data reduction rules leading to a polynomial-size problem kernel.

\begin{rrule}\label{rr:removeNonCoreVertex}
	Let~$(G=(V,E),k,r,\degList)$ be an instance of \DCE[e$^+$] and let $C \subseteq V$ be an \typeSet in~$G$.
	Then, safely remove all vertices in~$V \setminus C$.
\end{rrule}

We next show that we can apply \cref{rr:removeNonCoreVertex} in linear time.
Note that in our setting linear time does \emph{not} mean~$O(n+m)$ as it is usually the case with graph problems.
The reason is that the degree list function~$\degList$ may contain up to~$r$ possible degrees for each vertex; this gives up to~$rn \notin O(n+m)$ possible degrees overall.
Therefore, linear time means in our setting~$O(m + |\degList|)$ time, where~$|\degList| \ge n$ denotes the encoding size of~$\degList$.
In the following we will assume that the encoding of~$\degList$ requires at least one bit per possible degree and thus~$|\degList| \ge n$.

\begin{lemma} \label{lem:rruleSafe}
	\cref{rr:removeNonCoreVertex} is correct and can be applied in linear time.
\end{lemma}
\begin{proof}
	We first prove the correctness of \cref{rr:removeNonCoreVertex}. 
	To this end, the given \DCE[e$^+$] instance is denoted by~$I := (G=(V,E),k,r,\degList)$.
	We fix any \typeSet~$C \subseteq V$. 
	Furthermore, denote by~$I'$ the resulting instance when safely removing~$V \setminus C$, formally, $I' := (G[C], k, r, \degList_{V \setminus C})$.
	As all vertices in~$V \setminus C$ are satisfied, it follows that any edge set that is a solution for~$I'$ is also a solution for~$I$. 
	Hence, if~$I'$ is a yes-instance, then also~$I$ is a yes-instance. 
	To complete the correctness proof, it remains to prove the reverse direction.
	
	Let~$E' \subseteq \binom{V}{2} \setminus E$ be a solution for~$I$, that is, $\forall v \in V \colon \deg_{G+E'}(v) \in \degList(v)$.
	Observe that if~$V(E') \subseteq C$, then~$E'$ is also a solution for~$I'$.
	Hence, it remains to consider the case~$V(E') \setminus C \neq \emptyset$.
	Let~$v \in V(E') \setminus C$.
	We show how to construct from~$E'$ a solution~$E''$ for~$I$ such that~$(V(E') \setminus C) \setminus V(E'') = \{v\}$.
	Let~$i \le k$ denote the number of edges in~$E'$ with endpoint~$v$.
	Since~$v$ is not in the \typeSet~$C$, it follows that~$v \notin \mathcal{U}$ and~$|C \cap T_i| = \alpha = k(\Delta_G + 2)$.
	Next, we show that there is a type-$i$ vertex~$u\in C$ such that~$u \notin V(E')$ and~$u \notin N_G(N_{G[E']}(v))$, that is, $u$ is not incident to any edge in~$E'$ and also not adjacent to any vertex that is connected to~$v$ by an edge in~$E'$.
	Note that ``replacing''~$v$ by such a vertex~$u$ in the edge set~$E'$ yields~$E''$:
	Formally, for~$E'' := \{\{u,w\} \mid \{v,w\} \in E'\} \cup \{\{w_1,w_2\} \mid \{w_1,w_2\} \in E' \wedge w_1 \neq v \wedge w_2 \neq v\}$, it holds that~$E'' \cap E = \emptyset$ and since~$u$ is also of type~$i$, all degree constraints are satisfied in~$G+E''$.
	Hence, it remains to show that such a vertex~$u$ exists, that is, $(C \cap T_i) \setminus (V(E') \cup N_G(N_{G[E']}(v)))$ is indeed non-empty.
	This is true since~$|C \cap T_i| = k(\Delta_G +2)$, whereas $|V(E') \cup N_G(N_{G[E']}(v))| < 2k + k\Delta_G$.
	By iteratively applying this procedure, we obtain a solution for~$I'$. 
	Hence, $I'$~is a yes-instance if~$I$ is a yes-instance.
	This completes the correctness proof.

	To compute the \typeSet~$C$ in linear time, initialize~$C := \emptyset$ and $r$~counters~$c_1 := c_2 := \ldots := c_r := 0$ (one for each type).
	Then, for each vertex~$v$, compute the types of~$v$ in~$O(|\degList(v)|)$ time and let~$I \subseteq \{1, \ldots, r\}$ be the set of types of~$v$.
	If~$v$ is unsatisfied or if~$c_i \le \alpha$ for some~$i\in I$, then add~$v$ to~$C$.
	If~$v$ is satisfied, then increase~$c_i$ by one for each~$i \in I$.
	Now that we computed the vertices in~$C$ in linear time, it remains to compute their correct degree lists.
	To this end, for each vertex~$v \in C$, compute~$\gamma := \deg_G(v) - \deg_{G[C]}$ (doable in $O(\deg(v))$ time) and set~$\degList_{V \setminus C} (v) := \{d \ge 0 \mid d + \gamma \in \degList(v)\}$ in~$O(|\degList (v)|)$ time.
	Overall, we safely removed all vertices in~$V \setminus C$ in linear time. 
\end{proof}
As each \typeSet contains at most~$\alpha$ satisfied vertices of each vertex type, it follows that after one application of \cref{rr:removeNonCoreVertex} the graph contains at most~$|C| = |\mathcal{U}| + r \alpha$ vertices. 
The number of unsatisfied vertices in an \typeSet can always be upper-bounded by~$|\mathcal{U}|\le 2k$ since we can increase the degrees of at most~$2k$ vertices by adding~$k$ edges. 
If there are more than~$2k$ unsatisfied vertices, then we return a trivial no-instance.
Thus, we end up with $|C| \le 2k + rk(\Delta_G +2)$.
To obtain a polynomial-size problem kernel with respect to the combined parameter~$(k,r)$, we need to bound the maximum vertex degree~$\Delta_G$.
However, this can easily be achieved:
Since we only allow edge additions, for each vertex~$v \in V$, we have~$\deg(v) \le \max \degList(v) \le r$.
Formalized as a data reduction rule, this reads as follows:

\begin{rrule}\label{rr:maxDegree}
	Let~$(G=(V,E),k,r,\degList)$ be an instance of \DCE[e$^+$]. 
	If~$G$ contains more than~$2k$ unsatisfied vertices or if there exists a vertex~$v \in V$ with~$\deg(v) > \max \degList(v)$, then return a trivial no-instance.
\end{rrule}
%
Having applied \cref{rr:maxDegree} once, it holds that~$\Delta_G \le r$.
Thus, one execution of \cref{rr:removeNonCoreVertex} yields a graph containing at most~$2k+rk(r + 2)$ vertices.
\cref{lem:rruleSafe} ensures that we can apply \cref{rr:removeNonCoreVertex} in linear time. 
Clearly, \cref{rr:maxDegree} can be applied in linear time, too.
This leads to the following.

\begin{theorem}\label{thm:krpolykernelDCEa}
  \DCE[e$^+$] admits a problem kernel containing~$O(kr^2)$ vertices computable in linear time.
\end{theorem}
Having computed the problem kernel due to \cref{thm:krpolykernelDCEa}, one can solve \DCE[e$^+$] by simply trying all possibilities to add at most~$k$ edges with endpoints in the remaining~$O(k r^2)$ vertices.
This gives the following.
\begin{corollary}\label{cor:DCEe+fpt-rk}
	\DCE[e$^+$] can be solved in~$(kr)^{O(k)} + O(m + |\degList|)$ time.
\end{corollary}
      
\subsection{A Polynomial Kernel for \DCE[e$^+$] with Respect to $r$}\label{sec:param-r}
In this subsection, by adapting among other things some ideas of~\citet{HNNS15}, we show how to upper-bound~$k$ by a polynomial in~$r$.
Combining this upper bound for~$k$ with \cref{thm:krpolykernelDCEa} results in a polynomial-size problem kernel for the single parameter~$r$.
The general strategy to obtain the upper bound is inspired by a heuristic of \citet{LT08} and will be as follows:
First, remove the graph structure and solve the problem on the degree sequence of the input graph by using dynamic programming.
The solution to this number problem will indicate the \emph{demand}
for each vertex, that is, the number of added edges incident to that vertex.
Then, using a result due to \citet{KT00}, we prove that either~$k \le \realizationSizeBound$ or we can find a set of edges satisfying the specified demands in polynomial time.

We start by formally defining the number problem and showing its polynomial-time solvability.
\probDefEnv{\defDecprob{Number Constraint Editing (\NCE)}
{A function~$\numList\colon \{1,\ldots,n\} \rightarrow 2^{\{0,\dotsc,r\}}$ and positive integers $d_1, \ldots, d_n, k, r$.}
{Are there~$n$ positive integers~$d_1', \ldots, d_n'$ such that~$\sum_{i=1}^{n} (d_i' - d_i) = k$ and for all~$i \in \{1,\ldots,n\}$ it holds that~$d_i' \ge d_i$ and~$d_i' \in \numList(i)$?}}

\begin{lemma} \label{lem:polyNCE}
	\NCE is solvable in~$O(n\cdot k\cdot r)$ time.
\end{lemma}
\begin{proof}
	We provide a simple dynamic programming algorithm for \NCE.
	To this end, we define a two-dimensional table~$T$ as follows: 
	For~$i \in \{1,\ldots,n\}$ and~$j \in \{1,\ldots,k\}$, the entry~$T[i,j]$ is \texttt{true} if and only if the instance~$(d_1, \ldots, d_i, j, r, \numList)$ is a yes-instance.
	Hence, $T[n,k]$ stores the answer to the instance $(d_1,\ldots,d_n,k,r,\numList)$.

	In order to compute~$T$, we use the following recurrence: 
	\begin{align}
	(T[i,j] = \texttt{true}) \iff (\exists x \in \numList(i)\colon x \ge d_i \wedge T[i-1,j-(x-d_i)] = \texttt{true}), \label{eq:dpRecursion}
	\end{align}
	where we set
	$$T[1,j] := \begin{cases}
	            	\texttt{true,}  & \text{if } d_1 + j \in \numList(1), \\
	            	\texttt{false,} & \text{else.}
	            \end{cases}$$
	The correctness follows from the fact that at position~$i$ all possibilities for~$d_i'$ are considered.
	Also the running time is not hard to see: 
	There are~$n \cdot k$ entries and the computation of one entry requires to check at most~$r$ possibilities for the value~$x$ in Equivalence~(\ref{eq:dpRecursion}).
	As each check is doable in~$O(1)$ time, the overall running time sums up to~$O(n\cdot k\cdot r)$.
\end{proof}
\cref{lem:polyNCE} can be proved with a dynamic program that specifies
the demand for each vertex, that is, the number of added edges
incident to each vertex.
Given these demands, the remaining problem is to decide whether there exists a set of edges that satisfy these demands and are not contained in the input graph~$G$.
This problem is closely related to the polynomial-time solvable \fFactor problem~\cite[Chapter 10]{LP86}, a special case of \DCE[e$^-$] where~$|\degList(v)| = 1$ for all~$v\in V$; it is formally defined as follows:
\probDefEnv{\defDecprob{\fFactor}
{A graph~$G=(V,E)$ and a function~$f\colon V \rightarrow \N_0$.}
{Is there an \emph{\ffactor{}}, that is, a subgraph~$G' = (V, E')$ of~$G$ such that~$\deg_{G'}(v) = f(v)$ for all~$v \in V$?}}
Observe that our problem of satisfying the demands of the vertices
in~$G$ is essentially the question whether there is an \ffactor in the
complement graph~$\overline{G}$ where the function~$f$ stores the
demand of each vertex.
Having formulated our problem as \fFactor, we use the following result about the existence of an \ffactor.

\begin{lemma}[\citet{KT00}] \label{thm:generalfFactorDegreeCondition}
	Let $G = (V,E)$ be a graph with minimum vertex degree~$\delta_G$ and let $a \le b$ be two positive integers. 
	Suppose further that $$\delta_G \ge \frac{b}{a+b} |V|\text{ and } |V| > \frac{a+b}{a}(b+a-3).$$
	Then, for any function $f\colon V\rightarrow \{a, a + 1, . . . , b\}$ where $\sum_{v \in V} f(v)$ is even, $G$~has an \ffactor.
\end{lemma}

As we are interested in an \ffactor of the complement graph~$\overline{G}$ of our input graph~$G$ and the demand for each vertex is at most~$r$, we use \cref{thm:generalfFactorDegreeCondition} with $\delta_{\overline{G}}\ge n-r-1$, $a=1$, and $b = r$ yielding the following.

\begin{lemma} \label{cor:fFactorDegreeCondition}
	Let~$G = (V, E)$ be a graph with $n$ vertices, $\delta_G \ge n - r - 1$, $r \ge 1$, and let~$f\colon V\rightarrow\{1, \ldots, r\}$ be a function such that~$\sum_{v \in V} f(v)$ is even.
	If~$n \ge \factorSizeBound$, then~$G$ has an \ffactor.
\end{lemma}%
\begin{proof}
	Set~$a = 1$ and~$b = r$. Then, $\delta_G \ge n-r-1 \ge \frac{b}{a+b} n =
        \frac{r}{r+1} n$ holds if~$n \ge \factorSizeBound$, which is
        true by assumption.
 	Also, $\frac{a+b}{a}(b+a-3) = (r + 1)(r-2) = r^2 - r - 2 <
        \factorSizeBound \le n$ holds, and thus all conditions of \cref{thm:generalfFactorDegreeCondition} are fulfilled.
\end{proof}
We now have all ingredients to show that we can upper-bound~$k$ by~$\realizationSizeBound$ or solve the given instance of \DCE[e$^+$] in polynomial time.
The main technical statement towards this is the following.
\begin{lemma}\label{lem:DegreeFactorRealization}
	Let~$I := (G=(V,E),k,r,\degList)$ be an instance of \DCE[e$^+$] with~$k \ge \realizationSizeBound$ and~$V = \{v_1, \ldots, v_n\}$.
	If there exists a~$k' \in \{\realizationSizeBound, \ldots,k\}$ such that~$(\deg(v_1), \ldots, \deg(v_n),2k',r,\numList)$ with~$\numList(i) := \degList(v_i)$ is a yes-instance of~\NCE, then~$I$ is a yes-instance of \DCE[e$^+$].
\end{lemma}
\begin{proof}
	Assume that $(\deg(v_1),\ldots,\deg(v_n),2k',r,\numList)$ is a yes-instance of \NCE. 
	Let~$d_1',\ldots,d_n'$ be integers such that~$d_i' \in \degList(v_i)$, $\sum_{i=1}^{n} d_i' - \deg(v_i) = 2k'$, and~$d_i' \ge d_i$.
	Hence, we know that the degree constraints can numerically be satisfied, giving rise to a new target degree~$d_i'$ for each vertex~$v_i$. 
	Let~$A := \{v_i \in V \mid d_i' > \deg(v_i)\}$ denote the set of \emph{affected} vertices containing all vertices which require addition of at least one edge in order to fulfill their degree constraints.
	It remains to show that the degree sequence of the affected vertices can in fact be realized by adding~$k'$ edges to~$G[A]$. 
	To this end, it is sufficient to prove the existence of an $f$-factor in the complement graph~$\overline{G[A]}$ with~$f(v_i) := d_i' - \deg(v_i) \in \{1,\ldots,r\}$ for all~$v_i \in A$ since such an $f$-factor contains exactly the~$k'$ edges we want to add to~$G$. 
	Thus, it remains to check that all conditions of \cref{cor:fFactorDegreeCondition} are indeed satisfied to conclude the existence of the sought $f$-factor. 
	First, note that~$\delta_{\overline{G[A]}} \ge |A| -r -1$ since $\Delta_{G[A]} \le r$. 
	Moreover, $\sum_{v_i\in A} (d_i' - \deg(v_i)) = 2k' \le |A|r$, and thus~$|A| \ge 2k'/r \ge 2\factorSizeBound$. 
	Finally, $\sum_{v_i\in A}f(v_i) = 2k'$ is even and thus \cref{cor:fFactorDegreeCondition} applies. 
\end{proof}
As \NCE is polynomial-time solvable, \cref{lem:DegreeFactorRealization} states a \emph{win-win situation}: either the solution is bounded in size or can be found in polynomial time.
From this and \cref{thm:krpolykernelDCEa}, we obtain the polynomial-size problem kernel.

\begin{theorem}\label{thm:rpolykernelDCEa}
	\DCE[e$^+$] admits a problem kernel containing~$O(r^5)$ vertices computable in~$O(k^2 \cdot r \cdot n + m + |\tau|)$ time.
\end{theorem}
\begin{proof}
	Let $I := (G,k,r,\degList)$ be an instance of \DCE[e$^+$].
	We distinguish two cases concerning the size of~$k$.
	\begin{description}
		\item[Case 1. $k > \realizationSizeBound$:] 
		We solve for all~$k' \in \{\realizationSizeBound, \ldots,k\}$ the corresponding \NCE formulation.
		If for one~$k'$ we encounter a yes-instance of the \NCE formulation, then, justified by \cref{lem:DegreeFactorRealization}, we return a trivial yes-instance of constant size.
		By \cref{lem:polyNCE}, this can be done in polynomial time.    
		Otherwise, as each solution for \DCE[e$^+$] can be transferred to a solution of \NCE, it follows that there is no solution for~$I$ of size~$k'$ for any~$k' \in \{\realizationSizeBound, \ldots,k\}$. 
		Thus, $I$ is a yes-instance if and only if~$(G,\realizationSizeBound,r,\degList)$ is a yes-instance
		Hence, set~$k := \realizationSizeBound$ and proceed as in the Case~2.

		\item[Case 2. $k \le \realizationSizeBound$:] 
		We simply run the kernelization algorithm from \cref{thm:krpolykernelDCEa} on~$I$ to obtain an~$O(r^5)$-vertex problem kernel. 
	\end{description}
        Concerning the running time, observe that we have to solve at most~$k$ times an instance of \NCE.
	By \cref{lem:polyNCE}, we can determine in~$O(k \cdot r \cdot n)$ time for each of these at most~$k$ instances whether it is a yes- or no-instance.
	If one instance is a yes-instance, due to \cref{lem:DegreeFactorRealization}, then the kernelization algorithm can return a trivial yes-instance in constant time.
	Otherwise, we apply \cref{thm:krpolykernelDCEa} in~$O(m + |\tau|)$ time.
	Overall, this gives a running time of~$O(k^2 \cdot r \cdot n + m + |\tau|)$.
\end{proof}

Due to the bound on~$k$ given by \cref{lem:DegreeFactorRealization}, we can infer from \cref{thm:rpolykernelDCEa} and \cref{cor:DCEe+fpt-rk} the following.

\begin{corollary}
	\DCE[e$^+$] can be solved in~$r^{O(r^3)} + O(k^2 \cdot r \cdot n +m + |\degList|)$ time.
\end{corollary}

\section{A General Approach for Degree Sequence Completion} \label{sec:genFramework}

In the previous section, we dealt with the problem \DCE[e$^+$], where one only has to \emph{locally} satisfy the degree of each vertex.
In this section, we show how the presented ideas for \DCE[e$^+$] can also be used to solve more \emph{globally} defined problems where the degree sequence of the solution graph~$G'$ has to fulfill a given property.
For example, consider the problem of adding a minimum number of edges to obtain a regular graph, that is, a graph where all vertices have the same degree. 
In this case the degree of a vertex in the solution is a priori not known but depends on the degrees of the other vertices. 
Using \fFactor, this particular problem can be solved in polynomial time; however, there are many NP-hard problems of this kind, including \kDegAnon as will be discussed in \cref{sec:Application}.

The \emph{degree sequence} of a graph~$G = (V,E)$ with $n$~vertices is the $n$-tuple containing the vertex degrees in nonincreasing order.
Then, for some tuple property~$\Pi$, we consider the following problem:
\probDefEnv{\defDecprob{$\Pi$-Degree Sequence Completion (\PiEA)}
{A graph~$G=(V,E)$, an integer~$k \in \N$.}
{Is there a set of edges~$E' \subseteq \binom{V}{2} \setminus E$ with $|E'| \le k$ such that the degree sequence of~$G + E'$ fulfills~$\Pi$?} 
}


Note that \PiEA is not a generalization of \DCE[e$^+$] since in \DCE[e$^+$] one can require for two vertices~$u$ and~$v$ of the same degree that~$u$ gets two more incident edges and~$v$ not.
This cannot be expressed in \PiEA.
We remark that the results stated in this section can be extended to hold for a generalized version of \PiEA where a ``degree list function''~$\degList$ is given as additional input and the vertices in the solution graph~$G'$ also have to satisfy~$\degList$, thus generalizing \DCE[e$^+$].
For simplicity, however, we stick to the easier problem definition as stated above.

\subsection{Fixed-Parameter Tractability of $\Pi$-DSC}
In this subsection, we first generalize the ideas behind \cref{thm:krpolykernelDCEa} to show fixed-parameter tractability of \PiEA with respect to the combined parameter~$(k,\Delta_G)$.
Then, we present an adjusted  version of \cref{lem:DegreeFactorRealization} and apply it to show fixed-parameter tractability for \PiEA with respect to the parameter~$\Delta_{G'}$.
Clearly, a prerequisite for both these results is that the following problem has to be fixed-parameter tractable with respect to the parameter~$\Delta_T:=\max\{d_1,\ldots,d_n\}$.
\probDefEnv{\defDecprob{\PiDec}
{An integer tuple~$T = (d_1, \ldots, d_n)$.}
{Does~$T$ fulfill~$\Pi$?}}

For the next result, we need some definitions. 
For $0\le d\le \Delta_G$, let~$D_G(d) := \{v \in V \mid \deg_G(v)=d\}$ be the \emph{block} of degree~$d$, that is, the set of all vertices with degree~$d$ in~$G$.
A subset~$V' \subseteq V$ is an \emph{\blockSet} if~$V'$ contains for every~$d \in \{0, \ldots, \Delta_G\}$ exactly~$\min \{\alpha, |D_G(d)|\}$ vertices.
Recall that~$\alpha = k(\Delta_G + 2)$ (see \cref{sec:polyKernel-rk}), and notice the similarity of \blockSet{s} and \typeSet{s} (see \Cref{def:alpha-type-set}).
This similarity is not a coincidence as we use ideas of \cref{rr:removeNonCoreVertex} and \cref{lem:rruleSafe} to obtain the following lemma.

\begin{lemma}\label{lem:weakKernelSet}
	Let~$I := (G = (V,E),k)$ be a yes-instance of \PiEA and let~$C \subseteq V$ be an \blockSet. 
	Then, there exists a set of edges~$E' \subseteq \binom{C}{2} \setminus E$ with $|E'| \le k$ such that the degree sequence of~$G + E'$ fulfills~$\Pi$.
\end{lemma}
\begin{proof}
	Let~$I := (G = (V,E),k)$ be a yes-instance of \PiEA and let~$C
        \subseteq V$ be an \blockSet.
	Thus, there exists a set of edges~$E' \subseteq \binom{V}{2} \setminus E$ with $|E'| \le k$ such that the degree sequence~$\mathcal{D} = (d_1', \ldots, d_n')$ of~$G' := G + E'$ fulfills~$\Pi$.
	If~$V(E') \subseteq C$, then there is nothing to prove.
	Hence, assume that there exists a vertex~$v \in V(E') \setminus C$.
	We show how to construct from~$E'$ an edge set~$E''$ for~$I$ such that~$(V(E') \setminus C) \setminus V(E'') = \{v\}$ and the degree sequence of~$G'' := G + E''$ equals~$\mathcal{D}$.
	Since~$v$ is not in the \blockSet~$C$, it follows that~$|C \cap D_G(\deg_G(v))| = \alpha = k(\Delta_G + 2)$.
	Next, we prove that there is a vertex~$u\in D_G(\deg_G(v))$ such that~$u \notin V(E')$ and~$u \notin N_G(N_{G[E']}(v))$, that is, $u$~is not incident to any edge in~$E'$ and also not adjacent to any vertex that is connected to~$v$ by an edge in~$E'$.
	Note that ``replacing''~$v$ by such a vertex~$u$ in the edge set~$E'$ yields~$E''$:
	Formally, for $$E'' := \{\{u,w\} \mid \{v,w\} \in E'\} \cup \{\{w_1,w_2\} \mid \{w_1,w_2\} \in E' \wedge w_1 \neq v \wedge w_2 \neq v\},$$ it holds that~$E'' \cap E = \emptyset$ and since~$u \in D_G(\deg_G(v))$, the degree sequence of~$G+E''$ is~$\mathcal{D}$.
	Hence, it remains to show that such a vertex~$u$ exists, that is, $(C \cap D_G(\deg_G(v))) \setminus (V(E') \cup N_G(N_{G[E']}(v)))$ is indeed non-empty.
	This is true since~$|C \cap D_G(\deg_G(v))| = k(\Delta_G +2)$, whereas $|V(E') \cup N_G(N_{G[E']}(v))| < 2k + k\Delta_G$.
	By repeatedly applying this procedure, we obtain a solution~$E''' \subseteq \binom{C}{2} \setminus E$ with $|E'''| \le k$ such that the degree sequence of~$G + E'''$ fulfills~$\Pi$.
\end{proof}
In the context of \DCE, we introduced the notion of safely removing a vertex subset to obtain a problem kernel.
On the contrary, in the context of \PiEA, it seems impossible to
remove vertices in general without further knowledge about the tuple property~$\Pi$.
Thus, \cref{lem:weakKernelSet} does not lead to a problem kernel but only to a reduced search space for a solution, namely any \blockSet.
Clearly, an \blockSet~$C$ can be computed in polynomial time. 
Then, one can simply try out all possibilities to add edges with endpoints in~$C$ and check whether in one of the cases the degree sequence of the resulting graph satisfies~$\Pi$.
As~$|C| \le (\Delta_G + 2)k(\Delta_G+1)$, there are at
most~$O(2^{((\Delta_G + 2)k(\Delta_G+1))^2})$ possible subsets of edges to add.
Altogether, this leads to the following theorem.

\begin{theorem}\label{thm:PiEAfptDeltak}
	Let~$\Pi$ be some tuple property.
	If \PiDec is fixed-parameter tractable with respect to the
        maximum tuple entry~$\Delta_T$, then \PiEA is fixed-parameter tractable with respect to~$(k,\Delta_G)$.
\end{theorem}

\paragraph{Bounding the Solution Size~$k$ in~$\Delta_{G'}$}
We now show how to extend the ideas of \cref{sec:param-r} to the
context of \PiEA in order to bound the solution size~$k$ by a polynomial in~$\Delta_{G'}$.
The general procedure still is the one inspired by \citet{LT08}: 
Solve the number problem corresponding to \PiEA on the degree sequence of the input graph and then try to ``realize'' the solution.
To this end, we define the corresponding number problem as follows:
\probDefEnv{\defDecprob{$\Pi$-Number Sequence Completion (\PiNA)}
{Positive integers~$d_1, \ldots, d_n, k, \Delta$.}
{Are there~$n$ nonnegative integers~$x_1, \ldots, x_n$ with~$\sum_{i=1}^{n} x_i = k$ such that~$(d_1+x_1,\ldots,d_n+x_n)$ fulfills~$\Pi$ and $d_i+x_i \le \Delta$?}}
%
With these problem definitions, we can now generalize \cref{lem:DegreeFactorRealization}.

\begin{lemma}\label{lem:generalDegreeFactorRealization}
	Let~$I := (G,k)$ be an instance of~\PiEA with~$V = \{v_1, \ldots, v_n\}$ and~$k \ge \Delta_{G'}(\Delta_{G'}+1)^2$.
	If there exists a~$k' \in \{\Delta_{G'}(\Delta_{G'}+1)^2, \ldots,k\}$ such that the corresponding \PiNA instance~$I':=(\deg(v_1),\ldots,\deg(v_n), 2k', \Delta_{G'})$ is a yes-instance, then~$I$ is a yes-instance.
\end{lemma}
\begin{proof}
	Let~$I':=(\deg(v_1),\ldots,\deg(v_n), 2k',\Delta_{G'})$ with~$k' \in \{\Delta_{G'}(\Delta_{G'}+1)^2, \ldots,k\}$ be a yes-instance of \PiNA and let~$x_1, \ldots, x_n$ denote a solution for~$I'$.
	Defining the function~$f\colon V \to \N$ as~$f(v_i) := x_i$, we now prove that~$\overline{G}$ contains an \ffactor which forms a solution~$E'$ for~$I$.
	Denote by~$A$ the set of affected vertices, formally, $A := \{v_i \in V \mid 0 < x_i\}$.
	Observe that~$|A| \ge 2k' / \Delta_{G'} \ge 2(\Delta_{G'}+1)^2$ as~$k' \ge \Delta_{G'}(\Delta_{G'}+1)^2$.
	Furthermore, as the maximum degree~$\Delta_G$ in~$G$ is upper-bounded by$\Delta_{G'}$, it follows that~$\overline{G[A]}$ has minimum degree at least~$|A| - \Delta_{G'} - 1$.
	Finally, observe that~$f(v_i)\in \{1,\ldots,\Delta_{G'}\}$ for
        each~$v_i \in A$ and that~$\sum_{v_i\in A}f(v_i)=2k'$ is even.
        Hence, by \cref{cor:fFactorDegreeCondition}, $\overline{G[A]}$ contains an \ffactor.
	Thus, $\overline{G}$ also contains an \ffactor~$G' = (V, E')$
        and since~$(\deg(v_1)+x_1,\ldots,\deg(v_n)+x_n)$
        fulfills~$\Pi$, it follows that~$E'$ is a solution for~$I$, implying that~$I$ is a yes-instance.
\end{proof}
Let function~$g(|I|)$ denote the running time for solving the \PiNA instance~$I$.
Clearly, if there is a solution for an instance of \PiEA, then there
also exists a solution for the corresponding \PiNA instance. It
follows that we can decide whether there is a large solution for \PiEA
(adding at least~$\Delta_{G'}(\Delta_{G'}+1)^2$ edges) in~$k \cdot g(n \log(n))$ time.
Hence, we arrive at the following win-win situation:

\begin{lemma}\label{lem:boundSolSizeInDelta}
        Let $g$ denote the running time for solving \PiNA.
	There is an algorithm running in~$g(n\log(n))\cdot n^{O(1)}$ time that
        given an instance~$I := (G,k)$ of \PiEA returns ``yes'' or ``no''
        such that if it answers ``yes'', then~$I$ is a yes-instance, and otherwise
	$I$ is a yes-instance if and only if~$(G,\min\{k, \Delta_{G'}(\Delta_{G'}+1)^2\})$ is a yes-instance.
\end{lemma}
Using \cref{lem:boundSolSizeInDelta}, we can transfer the
fixed-parameter tractability for \PiNA with respect to~$\Delta$ to a
fixed-parameter tractability result for \PiEA with respect to~$\Delta_{G'}$.
Note that~$\Delta_{G'} \le k + \Delta_G$, that is, $\Delta_{G'}$ is
a smaller and thus ``stronger'' parameter~\cite{KN12}.
Also, showing \PiNA to be fixed-parameter tractable with respect
to~$\Delta$ might be a significantly easier task than proving
fixed-parameter tractability for \PiEA with respect to~$\Delta_{G'}$
directly since the graph structure can be completely ignored.

\begin{theorem} \label{thm:PiEA-fpt}
	If \PiNA is fixed-parameter tractable with respect to~$\Delta$, then \PiEA is fixed-parameter tractable with respect to~$\Delta_{G'}$.
\end{theorem}
\begin{proof}
    Let~$I:=(G,k)$ be a \PiEA instance. First, note that~$\Delta_G
    \le \Delta_{G'}$ always holds since we are only adding edges to~$G$.
    Thus, if~$k\le \Delta_{G'}(\Delta_{G'}+1)^2$, then the
    fixed-parameter tractability with respect to~$(k,\Delta_G)$ from
    \cref{thm:PiEAfptDeltak} yields fixed-parameter tractability with
    respect to~$\Delta_{G'}$. Otherwise, we use
    \cref{lem:boundSolSizeInDelta} to check whether there exists a large
    solution of size at least~$\Delta_{G'}(\Delta_{G'}+1)^2$. Hence,
    by assumption, in $f(\Delta_{G'})\cdot n^{O(1)}$~time for some
    computable function~$f$, we either
    find that~$I$ is a yes-instance or we can assume that~$k \le
    \Delta_{G'}(\Delta_{G'}+1)^2$, which altogether yields fixed-parameter
    tractability with respect to~$\Delta_{G'}$.
\end{proof}
If \PiNA can be solved in polynomial time, then \cref{lem:boundSolSizeInDelta} shows that we can assume that~$k \le \Delta_{G'}(\Delta_{G'}+1)^2$. 
Thus, as in the \DCE[e$^+$] setting (\autoref{thm:rpolykernelDCEa}), polynomial kernels with respect to~$(k,\Delta_G)$ transfer to the parameter~$\Delta_{G'}$, leading to the following.
\begin{theorem}\label{thm:PiEA-polyKernel}
	If \PiNA is polynomial-time solvable and \PiEA admits a
        polynomial kernel with respect to~$(k,\Delta_G)$, then \PiEA
        also admits a polynomial kernel with respect to~$\Delta_{G'}$.
\end{theorem}

\subsection{Applications}\label{sec:Application}
As our general approach is inspired by ideas of \citet{HNNS15},
it is not surprising that it can be applied to ``their'' \kDegAnon problem:

\probDefEnv{\defDecprob{\kDegAnon}
{An undirected graph $G = (V, E)$ and two positive integers $k$ and $s$.}
{Is there an edge set~$E'$ over~$V$ of size at most~$s$ such that~$G' := G + E'$
  is~$k$-\emph{anonymous}, that is, for each vertex~$v \in V$,
  there are at least~$k-1$ other vertices in~$G'$ having the same degree?}}
The property~$\Pi$ of being~$k$-anonymous clearly can be decided in polynomial time for a given degree sequence and thus, by \cref{thm:PiEAfptDeltak}, we immediately get fixed-parameter tractability with respect to~$(s,\Delta_G)$.
For example, \cref{thm:PiEA-polyKernel} then basically yields the kernel result obtained by~\citet{HNNS15}. 
There are more general versions of \kDegAnon as proposed by \citet{CKSV13}. 
For example, just a given subset of the vertices has to be anonymized or the vertices have labels.
As in each of these generalizations one can decide in polynomial time whether a given graph satisfies the particular anonymity requirement, \cref{thm:PiEAfptDeltak} applies also in these scenarios. 
However, checking in which of these more general settings the
conditions of \cref{thm:PiEA-fpt} or \cref{thm:PiEA-polyKernel} are fulfilled has to remain future work.

Besides the graph anonymization setting, one could think of further, more generalized constraints on the degree sequence. 
For example, if~$p_i(d)$ denotes how often degree~$i$ appears in a degree sequence~$\mathcal{D}$, then being~$k$-anonymous translates into~$p_i(\mathcal{D}_{G'}) \ge k$ for all degrees~$i$ occurring in the degree sequence~$\mathcal{D}_{G'}$ of the modified graph~$G'$.
Now, it is natural to consider not only a lower bound~$k \le p_i(\mathcal{D})$, but also an upper bound~$p_i(\mathcal{D}) \le u$ or maybe even a set of allowed frequencies~$p_i(\mathcal{D}) \in F_i \subseteq \N$.
Constraints like this allow to express some properties not of individual degrees itself but on the whole distribution of the degrees in the resulting sequence.
For example, to have some ``balancedness'' one can require that each occurring degree occurs exactly~$\ell$ times for some~$\ell \in \N$ \cite{CLMO89}.
To obtain some sort of ``robustness'' it might be useful to ask for an~$h$-index of~$\ell$, that is, in the solution graph there are at least~$\ell$ vertices with degree at least~$\ell$ \cite{ES12}.

Another range of problems which fit naturally into our framework involves completion problems to a graph class that is completely characterized by degree sequences. 
Many results concerning the relation between a degree sequence and the corresponding realizing graph are known and can be found in the literature; one of the first is the result by \citet{EG60} showing which degree sequences are in fact graphic, that is, realizable by a graph.
%
%
%
Based on this characterization researchers characterized for example pseudo-split, split, and threshold graphs completely by their degree sequences~\cite{HIS78,HS81,MP94}.
%
The NP-hard \textsc{Split Graph Completion}~\cite{NSS01} problem, for example, is known to be fixed-parameter tractable with respect to the allowed number of edge additions~\cite{Cai96}.
Note, however, that for the mentioned graph classes polynomial kernels with respect to the parameter~$\Delta_{G'}$ trivially exist because here we always have~$\sqrt{n} \le \Delta_G$. 

We finish with another interesting example of a class of graphs characterized by their degree sequence:
A graph is a \emph{unigraph} if it is determined by its degree sequence up to isomorphism~\cite{BLS99}. 
%
Given a degree sequence~$\mathcal{D} = (d_1, \ldots, d_n)$, one can decide in linear time whether~$\mathcal{D}$ defines a unigraph~\cite{BCP11,KL75}.
Again, by \cref{thm:PiEA-fpt}, we conclude fixed-parameter
tractability for the unigraph completion problem with respect to the parameter maximum degree in the solution graph~$\Delta_{G'}$.

\section{Conclusion}
We proposed a method for deriving efficient and effective preprocessing algorithms for degree sequence completion problems. \DCE[e$^+$] served as our main illustrating example. 
Roughly speaking, the core of the approach (as basically already used in previous work~\cite{LT08,HNNS15}) consists of extracting the degree sequence from the input graph, efficiently solving  a simpler number editing problem, and translating the obtained solution back into a solution for the graph problem using $f$-factors. 
While previous work~\cite{LT08,HNNS15} was specifically tailored towards an application for degree anonymization in graphs, we generalized the approach by filtering out problem-specific parts and ``universal'' parts. Thus, whenever one can solve these problem-specific parts efficiently, we can automatically obtain efficient preprocessing and fixed-parameter tractability results.

Our approach seems promising for future empirical investigations; an experimental work has already been performed for \kDegAnon~\cite{HHN14}. 
Another line of future research could be to study polynomial-time approximation algorithms for the considered degree sequence completion problems. 
Perhaps parts of our preprocessing approach might find use here as well.
A more specific open question concerning our work would be how to deal with additional connectivity requirements for the generated graphs.

\bibliographystyle{abbrvnat}
\bibliography{bibliography}

\end{document}